\documentclass{article}

\usepackage[margin=1in]{geometry}
\usepackage{amsmath,amssymb,amsthm,mathtools}
\usepackage{algorithm}
\usepackage{algpseudocode}
\usepackage{booktabs}
\usepackage{enumitem}
\usepackage[hidelinks]{hyperref}
\usepackage[font=small,labelfont=bf]{caption}

\newtheorem{theorem}{Theorem}
\newtheorem{lemma}[theorem]{Lemma}
\newtheorem{corollary}[theorem]{Corollary}
\newtheorem{definition}[theorem]{Definition}

\newcommand{\OPT}{\operatorname{OPT}}
\newcommand{\TMEC}{\textsc{TMEC}}
\newcommand{\TMNC}{\textsc{TMNC}}
\newcommand{\R}{\mathbb{R}}

\title{Threshold Minimum Cut with Terminal Quotas: Logarithmic and Planar Approximation Algorithms}
\author{Qi Duan, Carnegie Mellon University}
\date{}

\begin{document}

\maketitle

\begin{abstract}
We study threshold minimum cut problems with a distinguished root vertex, a set
of terminals, and a quota. In the threshold minimum edge cut problem
(\TMEC), the goal is to find a minimum-cost edge cut that disconnects at least
$k$ terminals from the root. In the threshold minimum node cut problem
(\TMNC), the goal is to delete a minimum-cost set of nonterminal, nonroot
vertices so that at least $k$ terminals become disconnected from the root. We
prove three approximation guarantees. First, undirected general-graph \TMEC{}
admits a randomized polynomial-time expected $O(\log n)$ approximation via a
R\"acke-style cut-dominating tree decomposition and an exact dynamic program on
trees. A standard repetition argument gives the same asymptotic ratio with high
probability. Second, planar \TMEC{} admits a factor-$2$ approximation by reducing
the threshold condition to planar weighted balanced cut. Third, bounded-degree
planar \TMNC{} admits a $2\Delta$-approximation, where $\Delta$ is the maximum
degree of a deletable vertex, by reducing the node-cost problem to the planar
edge-cut problem on the same graph. The results separate exact-quota guarantees
from bicriteria small-set-expansion-type guarantees and identify the
unbounded-degree planar node-cut case as the main remaining obstacle.
\end{abstract}

\section{Introduction}

Cut problems are among the central objects in graph algorithms and network
optimization. Classical minimum $s$--$t$ cut asks for the cheapest edge or
vertex set separating one specified terminal from another. Many applications,
however, require a more aggregated form of separation: one wants to isolate not
a single target but a sufficiently large number of targets from a protected
root. This motivates threshold minimum cut problems.

In this paper we study rooted threshold cut problems. The input consists of an
undirected graph $G=(V,E)$, a root vertex $s$, a terminal set
$T\subseteq V\setminus\{s\}$, and a quota $k$ with $1\le k\le |T|$. A feasible
solution must disconnect at least $k$ terminals from the root. We consider both
edge and node deletion. In the edge version, the solution is an edge boundary
$\delta(X)$ of a set $X$ containing at least $k$ terminals and not containing
$s$. In the node version, the solution is an outer vertex boundary $N(X)$, and
$s$ and all terminals are undeletable.

These problems are natural in several settings. In network reliability, one may
ask for the minimum intervention that isolates a threshold number of vulnerable
hosts from a protected source. In image segmentation, planar grid-like graphs
often encode pixels or superpixels, and quota constraints can express the need
to separate a sufficiently large region. In cyber-security attack graphs, a
node cut corresponds to blocking or mitigating actions represented by attack
states, and a quota can model disconnecting a sufficient number of harmful
terminal states from an initial compromise state.

\paragraph{Our results.}
We prove the following approximation guarantees.
\begin{enumerate}[leftmargin=*]
    \item For undirected general graphs, \TMEC{} admits a randomized
    polynomial-time expected $O(\log n)$ approximation. The algorithm samples a
    cut-dominating decomposition tree, solves the threshold problem exactly on
    the tree by dynamic programming, and maps the resulting tree cut back to
    the original graph.
    \item For planar graphs, \TMEC{} admits a factor-$2$ approximation. The
    reduction assigns weight $2|T|$ to the root, weight $1$ to every terminal,
    and weight $0$ to all other vertices. A $b$-balanced cut with
    $b=k/(3|T|)$ is then exactly strong enough to force at least $k$ terminals
    outside the root component.
    \item For planar graphs in which every deletable vertex has degree at most
    $\Delta$, \TMNC{} admits a $2\Delta$-approximation. The proof assigns each
    edge $\{u,v\}$ the cost $\min\{c(u),c(v)\}$, applies the planar \TMEC{}
    approximation, and converts the returned edge cut back into a vertex
    separator by covering every crossing edge with a cheapest deletable
    endpoint.
\end{enumerate}

The edge-cut results preserve the quota exactly: the returned solution always
disconnects at least $k$ terminals. This distinction is important because
small-set-expansion-type algorithms often give stronger-looking bounds only in
a bicriteria sense, for example by returning a cut that separates
$\Omega(k)$ terminals rather than at least $k$ terminals.

\section{Related Work}

Balanced separator and sparsest-cut approximations are classical tools in graph
partitioning. Leighton and Rao gave logarithmic approximations based on
multicommodity flow and metric embeddings~\cite{LeightonRao1999}. Arora, Rao,
and Vazirani improved the approximation ratio for sparsest cut and balanced
separator to $O(\sqrt{\log n})$ using semidefinite programming and geometric
rounding~\cite{AroraRaoVazirani2009}. These results are closely related in
spirit but do not directly enforce the exact rooted terminal quota considered
here.

R\"acke introduced hierarchical decompositions for congestion minimization and
oblivious routing, giving logarithmic-loss tree representations of general
graphs~\cite{Racke2008}. We use a standard cut-dominating formulation of this
line of work as a black box: each sampled tree dominates every graph cut, and
the expected tree value of every graph cut is within an $O(\log n)$ factor.
This makes it possible to reduce exact-quota edge-cut objectives on general
graphs to exact dynamic programming on trees.

For planar graphs, Garg, Saran, and Vazirani gave a factor-$2$ approximation
for minimum-cost $b$-balanced edge cut when $b\le 1/3$, assuming unary vertex
weights; with binary weights they give a pseudoapproximation
algorithm~\cite{GargSaranVazirani1999}. Our planar \TMEC{} result is a direct
quota-preserving reduction to this theorem.

Small-set expansion and min--max graph partitioning provide another family of
methods for small unbalanced cuts. Bansal et al. give bicriteria algorithms
with guarantees of the form $O(\sqrt{\log n\log(1/\rho)})$ for small-set
expansion and related partitioning problems~\cite{BansalEtAl2014}. These tools
suggest possible bicriteria improvements for threshold cuts, but the distinction
between separating $\Omega(k)$ terminals and separating at least $k$ terminals
is essential for the exact-feasibility results in this paper.

\section{Preliminaries}
\label{sec:preliminaries}

Let $G=(V,E)$ be an undirected graph. For a set $X\subseteq V$, its edge
boundary is
\[
\delta_G(X)
=
\{\{u,v\}\in E:u\in X,\ v\notin X\}.
\]
If the graph is clear from context, we write $\delta(X)$ instead of
$\delta_G(X)$. If edges have nonnegative costs $c:E\to \R_{\ge 0}$, then
\[
c_G(\delta_G(X))
=
\sum_{e\in \delta_G(X)} c(e).
\]
For vertex-deletion problems, the outer vertex boundary of $X$ is
\[
N_G(X)
=
\{v\in V\setminus X:\text{ there exists }u\in X\text{ with }\{u,v\}\in E\}.
\]
Again, we write $N(X)$ when the graph is clear.

\begin{definition}[Threshold minimum edge cut]
Given an undirected graph $G=(V,E)$ with nonnegative edge costs
$c:E\to \R_{\ge 0}$, a root $s\in V$, a terminal set
$T\subseteq V\setminus\{s\}$, and a quota $k$ with $1\le k\le |T|$, the
threshold minimum edge cut problem is
\[
\OPT_E(G,s,T,k)
=
\min
\left\{
    c_G(\delta_G(X)):
    X\subseteq V,\ s\notin X,\ |X\cap T|\ge k
\right\}.
\]
We abbreviate this problem as \TMEC.
\end{definition}

\begin{definition}[Threshold minimum node cut]
Given an undirected graph $G=(V,E)$ with nonnegative vertex deletion costs
$c:V\to \R_{\ge 0}\cup\{+\infty\}$, a root $s\in V$, a terminal set
$T\subseteq V\setminus\{s\}$, and a quota $k$ with $1\le k\le |T|$, the
threshold minimum node cut problem is
\[
\OPT_N(G,s,T,k)
=
\min
\left\{
    c(N_G(X)):
    X\subseteq V,\ s\notin X,\ |X\cap T|\ge k
\right\},
\]
where
\[
c(N_G(X))=\sum_{v\in N_G(X)}c(v).
\]
The root $s$ and terminals $T$ are undeletable, which we model by assigning
$c(s)=+\infty$ and $c(t)=+\infty$ for every $t\in T$. We abbreviate this problem
as \TMNC.
\end{definition}

\section{General-Graph Threshold Minimum Edge Cut}
\label{sec:general-tmec}

In this section we prove that \TMEC{} admits an $O(\log n)$ approximation on
undirected general graphs. The algorithm uses a cut-dominating tree
decomposition and an exact dynamic program on trees.

If $G$ is disconnected, then terminals already outside the connected component
of $s$ are disconnected at zero cost. Thus one may first count those terminals.
If at least $k$ terminals are already disconnected from $s$, the empty cut is
optimal. Otherwise, we restrict to the connected component of $s$ and reduce
the quota accordingly. Hence, in the sequel, we assume without loss of
generality that $G$ is connected.

\subsection{Cut-Dominating Tree Decomposition}

We use the following cut-dominating tree-decomposition theorem as a black box.
It is a standard formulation of R\"acke-style cut-based hierarchical
decompositions.

\begin{theorem}[Cut-dominating tree decomposition]
\label{thm:racke-cut-tree}
For every undirected edge-weighted graph $G=(V,E,c)$ with $n=|V|$, there is a
randomized polynomial-time construction of a weighted tree $H$ whose leaves are
in bijection with the vertices of $G$ such that the following holds. For every
set $A\subseteq V$, let $\lambda_H(A)$ denote the minimum cost of a set of tree
edges whose removal separates the leaves corresponding to $A$ from the leaves
corresponding to $V\setminus A$. Then, for every $A\subseteq V$,
\[
c_G(\delta_G(A))
\le
\lambda_H(A),
\]
and
\[
\mathbb{E}_{H}\left[\lambda_H(A)\right]
\le
\alpha(n)c_G(\delta_G(A)),
\]
where
\[
\alpha(n)=O(\log n).
\]
\end{theorem}

The first inequality is the domination property: every tree cut dominates the
corresponding graph cut. The second inequality says that the expected tree
overestimate of every graph cut is only logarithmic.

\subsection{Exact \TMEC{} on a Tree}

We next show that \TMEC{} can be solved exactly on a weighted tree whose leaves
represent the original graph vertices.

Let $H$ be a weighted tree whose leaves are in bijection with the original
vertices $V$. We root $H$ at the leaf corresponding to $s$. For a node $u$ of
$H$, let $H_u$ denote the rooted subtree of $H$ below $u$, and let
$L(u)\subseteq V$ be the set of original vertices whose leaves lie in $H_u$.
Define
\[
\tau(u)=|L(u)\cap T|.
\]
All terminal counts are truncated at $k$.

For every node $u$ and every $j\in\{0,1,\dots,k\}$, define $DP[u,j]$ to be the
minimum cost of cutting tree edges strictly inside $H_u$ so that exactly $j$
terminals in $L(u)\cap T$ are disconnected from $u$, under the condition that
$u$ remains connected to its parent. The value $j=k$ means ``at least $k$''
after truncation.

For a rooted leaf $u$, no edge lies strictly inside $H_u$, so
\[
DP[u,0]=0,
\]
and
\[
DP[u,j]=+\infty
\qquad\text{for all }j>0.
\]
The root of the oriented tree is the original leaf corresponding to $s$; it is
treated by the same recurrence using its children. Thus, although it is a leaf
of the unrooted decomposition tree, it is not a rooted leaf unless $H$ has one
node.

Now consider an internal rooted node $u$ with children
\[
u_1,u_2,\dots,u_d.
\]
For a child $u_i$, there are two choices.

\paragraph{Keep the edge $(u,u_i)$.}
If the edge is kept, then the number of disconnected terminals contributed by
the child subtree is determined by $DP[u_i,\cdot]$.

\paragraph{Cut the edge $(u,u_i)$.}
If the edge $(u,u_i)$ is cut, then all terminals in $L(u_i)\cap T$ are
disconnected from $u$. This contributes
\[
\min\{\tau(u_i),k\}
\]
terminals and costs
\[
c_H(u,u_i).
\]

Thus define the child contribution table $B_i$ by initially setting
\[
B_i[j]=DP[u_i,j]
\qquad\text{for }j=0,1,\dots,k,
\]
and then updating
\[
B_i[\min\{\tau(u_i),k\}]
=
\min
\left\{
    B_i[\min\{\tau(u_i),k\}],
    c_H(u,u_i)
\right\}.
\]

The table $DP[u,\cdot]$ is obtained by knapsack convolution over the children.
Initialize
\[
Q_0[0]=0
\]
and
\[
Q_0[j]=+\infty
\qquad\text{for }j>0.
\]
For $i=1,\dots,d$, define
\[
Q_i[j]
=
\min_{\substack{a,b\in\{0,1,\dots,k\}\\ \min\{a+b,k\}=j}}
\left\{
    Q_{i-1}[a]+B_i[b]
\right\}.
\]
Finally set
\[
DP[u,j]=Q_d[j]
\qquad\text{for }j=0,1,\dots,k.
\]

At the root $r$, which corresponds to the original root vertex $s$, the optimum
tree value is
\[
\OPT_H=DP[r,k].
\]

\begin{lemma}[Exact tree algorithm]
\label{lem:tmec-tree-dp}
\TMEC{} on a weighted tree whose leaves are the original vertices can be solved
exactly in polynomial time. A straightforward implementation runs in
\[
O(|V(H)|k^2)
\]
time.
\end{lemma}

\begin{proof}
For every child subtree, an optimal tree solution has exactly two possibilities:
either it cuts the edge to the child, thereby disconnecting all terminals in
that child subtree, or it keeps the edge and recursively disconnects some
terminals inside the child subtree. Since a tree has no cycles, choices in
different child subtrees are independent once the state at the parent is fixed.

The recurrence enumerates all such possibilities and combines them by knapsack
convolution. By induction from the rooted leaves to the root, $DP[u,j]$ is
exactly the minimum cost of disconnecting $j$ terminals inside $H_u$ while
keeping $u$ connected to its parent. At the root, the state $j=k$ represents
disconnecting at least $k$ terminals from the root leaf corresponding to $s$.
Hence $DP[r,k]$ is the exact optimum tree value.
\end{proof}

\subsection{The Approximation Algorithm}

\begin{algorithm}[H]
\caption{General-graph \TMEC{} via cut-dominating trees}
\label{alg:general-tmec-racke}
\begin{algorithmic}[1]
\Require Undirected graph $G=(V,E)$, edge costs $c$, root $s$, terminals $T$, quota $k$.
\State Sample a cut-dominating tree $H$ using Theorem~\ref{thm:racke-cut-tree}.
\State Solve \TMEC{} exactly on $H$ using Lemma~\ref{lem:tmec-tree-dp}.
\State Let $F_H$ be the tree edge set returned by the dynamic program.
\State Let $X_H=\{v\in V:\text{ the leaf corresponding to }v\text{ is disconnected from }s\text{ in }H-F_H\}$.
\State Return the graph cut $\delta_G(X_H)$.
\end{algorithmic}
\end{algorithm}

\begin{theorem}[General-graph \TMEC{} approximation]
\label{thm:general-tmec-logn}
Threshold minimum edge cut on undirected general graphs admits a randomized
polynomial-time expected
\[
O(\log n)
\]
approximation.
\end{theorem}

\begin{proof}
Let $X^\star$ be an optimal \TMEC{} solution in $G$. Thus
\[
s\notin X^\star,
\qquad
|X^\star\cap T|\ge k,
\]
and
\[
c_G(\delta_G(X^\star))=\OPT_G,
\]
where $\OPT_G=\OPT_E(G,s,T,k)$.

Sample a cut-dominating tree $H$ from the distribution guaranteed by
Theorem~\ref{thm:racke-cut-tree}. Since $X^\star$ separates at least $k$
terminals from $s$, the minimum tree edge set separating the leaves
corresponding to $X^\star$ from the leaves corresponding to $V\setminus X^\star$
is feasible for the tree \TMEC{} instance. Therefore, if $F_H$ is the optimum
tree edge set found by the dynamic program, then
\[
c_H(F_H)
\le
\lambda_H(X^\star).
\]

By the expected distortion guarantee of Theorem~\ref{thm:racke-cut-tree},
\[
\mathbb{E}_{H}[\lambda_H(X^\star)]
\le
O(\log n)c_G(\delta_G(X^\star))
=
O(\log n)\OPT_G.
\]
Hence
\[
\mathbb{E}_{H}[c_H(F_H)]
\le
O(\log n)\OPT_G.
\]

Now define
\[
X_H
=
\{v\in V:\text{ the leaf corresponding to }v\text{ is disconnected from }s
\text{ in }H-F_H\}.
\]
Since $F_H$ disconnects at least $k$ terminal leaves from the root leaf, we
have
\[
s\notin X_H,
\qquad
|X_H\cap T|\ge k.
\]
Thus $\delta_G(X_H)$ is a feasible \TMEC{} solution in $G$.

It remains to bound its graph cost. By the domination property of
Theorem~\ref{thm:racke-cut-tree},
\[
c_G(\delta_G(X_H))
\le
\lambda_H(X_H).
\]
Furthermore, $F_H$ itself separates the leaves in $X_H$ from the leaves in
$V\setminus X_H$. Therefore,
\[
\lambda_H(X_H)
\le
c_H(F_H).
\]
Combining these inequalities gives
\[
c_G(\delta_G(X_H))
\le
c_H(F_H).
\]
Taking expectation,
\[
\mathbb{E}_{H}\left[c_G(\delta_G(X_H))\right]
\le
\mathbb{E}_{H}[c_H(F_H)]
\le
O(\log n)\OPT_G.
\]
Therefore Algorithm~\ref{alg:general-tmec-racke} is an expected
$O(\log n)$ approximation.
\end{proof}

\begin{corollary}[High-probability version]
\label{cor:general-tmec-high-prob}
For any $\eta\in(0,1)$, by sampling
\[
O(\log(1/\eta))
\]
independent decomposition trees, solving the tree instance on each, and
returning the cheapest feasible graph cut obtained, one gets an
\[
O(\log n)
\]
approximation with probability at least $1-\eta$.
\end{corollary}

\begin{proof}
If $\OPT_G=0$, the expected-cost proof of Theorem~\ref{thm:general-tmec-logn}
implies that each sampled solution has cost $0$ with probability $1$, so the
claim is immediate. Assume $\OPT_G>0$. For one sampled tree,
Theorem~\ref{thm:general-tmec-logn} gives
\[
\mathbb{E}\left[c_G(\delta_G(X_H))\right]
\le
C\log n\cdot \OPT_G
\]
for some absolute constant $C$. By Markov's inequality,
\[
\Pr\left[
    c_G(\delta_G(X_H))>2C\log n\cdot \OPT_G
\right]
\le
\frac12.
\]
After $q$ independent repetitions, the probability that every sampled solution
has cost larger than $2C\log n\cdot \OPT_G$ is at most $2^{-q}$. Choosing
\[
q=\lceil \log_2(1/\eta)\rceil
\]
gives success probability at least $1-\eta$.
\end{proof}

\section{Planar Threshold Minimum Cut}
\label{sec:planar-threshold-cut}

In this section we prove a factor-$2$ approximation for planar \TMEC{} and a
$2\Delta$-approximation for planar \TMNC{} when every deletable vertex has
maximum degree at most $\Delta$.

\subsection{A Factor-$2$ Approximation for Planar \TMEC}

We reduce planar \TMEC{} to planar weighted balanced cut. We use the following
theorem of Garg, Saran, and Vazirani~\cite{GargSaranVazirani1999}.

\begin{theorem}[Planar balanced cut approximation]
\label{thm:gsv-balanced-cut}
Let $G=(V,E)$ be a planar graph with nonnegative edge costs and nonnegative
vertex weights $\mu:V\to \R_{\ge 0}$. For every balance parameter $b\le 1/3$,
there is a polynomial-time factor-$2$ approximation for the minimum-cost
$b$-balanced edge cut when the vertex weights are unary.
\end{theorem}

A cut is $b$-balanced if, after removing the cut edges, every connected
component has weight at most
\[
(1-b)\mu(V).
\]

\begin{theorem}[Planar \TMEC]
\label{thm:planar-tmec-2approx}
Let $G=(V,E)$ be a planar graph with nonnegative edge costs, root $s$, terminal
set $T$, and quota $k$. Then planar \TMEC{} admits a polynomial-time factor-$2$
approximation, assuming the unary-weight version of
Theorem~\ref{thm:gsv-balanced-cut}.
\end{theorem}

\begin{proof}
Let
\[
m=|T|.
\]
Define vertex weights $\mu$ by
\[
\mu(s)=2m,
\]
\[
\mu(t)=1
\qquad\text{for every }t\in T,
\]
and
\[
\mu(v)=0
\qquad\text{for every }v\in V\setminus(\{s\}\cup T).
\]
Then
\[
\mu(V)=3m.
\]
Set
\[
b=\frac{k}{3m}.
\]
Since $k\le m$, we have
\[
b\le \frac13.
\]
Therefore Theorem~\ref{thm:gsv-balanced-cut} applies. The weights are bounded
by $2m$ and hence are unary-expandable in polynomial size.

We first show that every feasible \TMEC{} solution induces a feasible
$b$-balanced cut. Let $X\subseteq V$ satisfy
\[
s\notin X
\qquad\text{and}\qquad
|X\cap T|\ge k.
\]
The side containing $s$ has weight at most
\[
2m+(m-k)
=
3m-k
=
(1-b)\mu(V).
\]
The side $X$ has weight at most
\[
m
\le
3m-k
=
(1-b)\mu(V).
\]
Thus $\delta_G(X)$ is a feasible $b$-balanced cut. Hence
\[
\OPT_{\mathrm{bal}}
\le
\OPT_E(G,s,T,k).
\]

Now let $F$ be the cut returned by the factor-$2$ planar balanced-cut
algorithm. Then
\[
c(F)
\le
2\OPT_{\mathrm{bal}}
\le
2\OPT_E(G,s,T,k).
\]
Let $R$ be the connected component containing $s$ in $G-F$, and define
\[
X=V\setminus R.
\]
Since $F$ is $b$-balanced, the root component $R$ satisfies
\[
\mu(R)\le 3m-k.
\]
Because $s\in R$, we have
\[
\mu(R)=2m+|R\cap T|.
\]
Therefore
\[
2m+|R\cap T|
\le
3m-k,
\]
which implies
\[
|R\cap T|\le m-k.
\]
Hence
\[
|X\cap T|\ge k.
\]
Also $s\notin X$. Thus $X$ is feasible for \TMEC{}.

Finally, every edge of $\delta_G(X)$ belongs to the removed edge set $F$, so
\[
c(\delta_G(X))
\le
c(F)
\le
2\OPT_E(G,s,T,k).
\]
Therefore planar \TMEC{} admits a factor-$2$ approximation.
\end{proof}

\subsection{A \texorpdfstring{$2\Delta$}{2-Delta}-Approximation for Bounded-Degree Planar \TMNC}

We now derive a constant-factor approximation for planar \TMNC{} when the
maximum degree of deletable vertices is bounded.

Let
\[
D=V\setminus(\{s\}\cup T)
\]
be the set of deletable vertices, and define
\[
\Delta
=
\max\{\deg(v):v\in D\}.
\]
We assume that the \TMNC{} instance has a finite feasible solution. If no
finite feasible solution exists, then the problem is infeasible under the
undeletable root and terminal constraints.

Let
\[
C_{\mathrm{all}}
=
\sum_{v\in D} c(v).
\]
Choose a sufficiently large number
\[
M>2\Delta C_{\mathrm{all}}.
\]
For notational convenience, set
\[
c(s)=M
\]
and
\[
c(t)=M
\qquad\text{for every }t\in T.
\]
For every edge $e=\{u,v\}\in E$, define the derived edge cost
\[
w(e)=\min\{c(u),c(v)\}.
\]
This gives a planar \TMEC{} instance on the same graph.

\begin{lemma}[Node optimum induces a bounded-cost edge cut]
\label{lem:node-to-edge-bounded}
Let $\OPT_N$ be the optimum value of the planar \TMNC{} instance, and let
$\OPT_E^w$ be the optimum value of the derived planar \TMEC{} instance with
edge costs $w(e)=\min\{c(u),c(v)\}$. Then
\[
\OPT_E^w
\le
\Delta\OPT_N.
\]
\end{lemma}

\begin{proof}
Let $S^\star$ be an optimal node separator for \TMNC{}, so
\[
c(S^\star)=\OPT_N.
\]
Let $R$ be the connected component containing $s$ in $G-S^\star$, and define
\[
X^\star=V\setminus(R\cup S^\star).
\]
Since $S^\star$ disconnects at least $k$ terminals from $s$, we have
\[
s\notin X^\star
\qquad\text{and}\qquad
|X^\star\cap T|\ge k.
\]
Therefore $X^\star$ is feasible for the derived edge-cut threshold problem.

Every edge in $\delta_G(X^\star)$ has one endpoint in $X^\star$ and the other
endpoint in $S^\star$. Otherwise, an edge from $X^\star$ to the root component
would make the two vertices belong to the same connected component of
$G-S^\star$. Hence, for every crossing edge $e=\{u,v\}$ incident to a separator
vertex $v\in S^\star$, we have
\[
w(e)\le c(v).
\]
Therefore
\[
w(\delta_G(X^\star))
\le
\sum_{v\in S^\star}\deg(v)c(v).
\]
Since every deletable vertex has degree at most $\Delta$,
\[
w(\delta_G(X^\star))
\le
\Delta\sum_{v\in S^\star}c(v)
=
\Delta\OPT_N.
\]
Thus
\[
\OPT_E^w\le \Delta\OPT_N.
\]
\end{proof}

\begin{lemma}[Edge cut induces a node separator]
\label{lem:edge-to-node-cover}
Let $X\subseteq V$ satisfy
\[
s\notin X
\qquad\text{and}\qquad
|X\cap T|\ge k.
\]
Suppose also that
\[
w(\delta_G(X))<M.
\]
Then one can construct in polynomial time a node separator $S_X\subseteq D$
such that at least $k$ terminals are disconnected from $s$ in $G-S_X$ and
\[
c(S_X)\le w(\delta_G(X)).
\]
\end{lemma}

\begin{proof}
Since
\[
w(\delta_G(X))<M,
\]
no edge in $\delta_G(X)$ can have both endpoints undeletable. Indeed, if an
edge had both endpoints in $\{s\}\cup T$, then its edge cost would be $M$,
forcing $w(\delta_G(X))\ge M$, a contradiction.

Therefore every crossing edge $e\in\delta_G(X)$ has at least one finite-cost
deletable endpoint. For every edge $e=\{u,v\}\in\delta_G(X)$, choose a deletable
endpoint of minimum cost and add it to $S_X$. Since
\[
w(e)=\min\{c(u),c(v)\},
\]
we have
\[
c(S_X)
\le
\sum_{e\in\delta_G(X)} w(e)
=
w(\delta_G(X)).
\]

The set $S_X$ hits every edge crossing between $X$ and $V\setminus X$. Hence,
after deleting $S_X$, no vertex of $X\setminus S_X$ can be connected to $s$
through a crossing edge. Moreover, $S_X\subseteq D$, so no terminal is deleted.
Since $|X\cap T|\ge k$, at least $k$ terminals are disconnected from $s$ in
$G-S_X$. Thus $S_X$ is a feasible \TMNC{} solution.
\end{proof}

\begin{theorem}[Bounded-degree planar \TMNC]
\label{thm:bounded-degree-planar-tmnc}
Let $G=(V,E)$ be a planar graph with nonnegative vertex deletion costs. Let
$s$ be an undeletable root, let $T$ be a set of undeletable terminals, and let
$k$ be a quota. Suppose the \TMNC{} instance has a finite feasible solution and
every deletable vertex has degree at most $\Delta$. Then planar \TMNC{} admits
a polynomial-time $2\Delta$-approximation.
\end{theorem}

\begin{proof}
Construct the derived planar \TMEC{} instance with edge costs
\[
w(\{u,v\})=\min\{c(u),c(v)\},
\]
where $c(s)=M$ and $c(t)=M$ for all $t\in T$, with
\[
M>2\Delta C_{\mathrm{all}}.
\]

By Theorem~\ref{thm:planar-tmec-2approx}, we can find a set $X\subseteq V$
such that
\[
s\notin X,
\qquad
|X\cap T|\ge k,
\]
and
\[
w(\delta_G(X))\le 2\OPT_E^w.
\]
By Lemma~\ref{lem:node-to-edge-bounded},
\[
\OPT_E^w\le \Delta\OPT_N.
\]
Therefore
\[
w(\delta_G(X))
\le
2\Delta\OPT_N.
\]
Since the instance has a finite feasible solution,
\[
\OPT_N\le C_{\mathrm{all}}.
\]
Thus
\[
w(\delta_G(X))
\le
2\Delta C_{\mathrm{all}}
<
M.
\]
Hence Lemma~\ref{lem:edge-to-node-cover} applies. It constructs a node
separator $S_X\subseteq D$ such that at least $k$ terminals are disconnected
from $s$ and
\[
c(S_X)\le w(\delta_G(X)).
\]
Consequently,
\[
c(S_X)
\le
w(\delta_G(X))
\le
2\Delta\OPT_N.
\]
Thus $S_X$ is a feasible \TMNC{} solution of cost at most
$2\Delta\OPT_N$.
\end{proof}

\section{Conclusion}

We introduced and analyzed threshold minimum cut problems with rooted terminal
quotas. For general undirected graphs, we gave an exact-feasibility
$O(\log n)$ approximation for threshold minimum edge cut using a cut-dominating
tree decomposition and an exact tree dynamic program. For planar graphs, we
proved a factor-$2$ approximation for threshold minimum edge cut by reducing the
quota condition to planar weighted balanced cut. Finally, for bounded-degree
planar graphs, we proved a $2\Delta$-approximation for threshold minimum node
cut by converting vertex deletion costs into edge costs and then applying the
planar edge-cut result.

The results leave several natural directions open. The most direct one is to
remove the factor $\Delta$ in the planar node-cut approximation, which would
likely require a direct approximation for planar balanced vertex separator or a
more faithful transformation from vertex costs to planar edge cuts. A second
direction is to improve the exact-quota general-graph edge-cut ratio below
$O(\log n)$. Existing small-set-expansion methods suggest better bicriteria
bounds, but exact preservation of the quota $|X\cap T|\ge k$ remains the key
technical obstacle.

\bibliographystyle{plain}
\bibliography{main}

\end{document}